\documentclass[preprint,12pt,authoryear]{elsarticle}

\usepackage{amssymb}

\journal{Finite Fields and Their Applications}

\begin{document}

\def\bbbr{{\rm I\!R}} %reelle Zahlen
\def\bbbm{{\rm I\!M}}
\def\bbbn{{\rm I\!N}} %natuerliche Zahlen
\def\bbbf{{\rm I\!F}}
\def\bbbh{{\rm I\!H}}
\def\bbbk{{\rm I\!K}}
\def\bbbp{{\rm I\!P}}
\def\bbbz{{\mathchoice {\hbox{$\sf\textstyle Z\kern-0.4em Z$}}
{\hbox{$\sf\textstyle Z\kern-0.4em Z$}}
{\hbox{$\sf\scriptstyle Z\kern-0.3em Z$}}
{\hbox{$\sf\scriptscriptstyle Z\kern-0.2em Z$}}}}

\newtheorem{definition}{Definition}
\newtheorem{proposition}{Proposition}
\newtheorem{theorem}{Theorem}
\newtheorem{lemma}{Lemma}
\newtheorem{corollary}{Corollary}
\newtheorem{remark}{Remark}
\newtheorem{example}{Example}
\newtheorem{conjecture}{Conjecture}
\newtheorem{acknowledgments}{Acknowledgments}

\newenvironment{proof}{\begin{trivlist}\item[]{\em Proof: }}{\samepage \hfill{\hbox{\rlap{$\sqcap$}$\sqcup$}}\end{trivlist}}

\begin{frontmatter}

\title{A characterization of a class of optimal three-weight cyclic codes of dimension 3 over any finite field\tnoteref{t1}}
\tnotetext[t1]{Partially supported by PAPIIT-UNAM IN107515.}

\author{Gerardo Vega}

\address{Direcci\'on General de C\'omputo y de Tecnolog\'{\i}as de Informaci\'on y Comunicaci\'on, Uni\-ver\-si\-dad Nacional Aut\'onoma de M\'exico, 04510 M\'exico D.F., Mexico, gerardov@unam.mx}

\begin{abstract}
It is well known that the problem of determining the weight distributions of families of cyclic codes is, in general, notoriously difficult. An even harder problem is to find characterizations of families of cyclic codes in terms of their weight distributions. On the other hand, it is also well known that cyclic codes with few weights have a great practical importance in coding theory and cryptography. In particular, cyclic codes having three nonzero weights have been studied by several authors, however, most of these efforts focused on cyclic codes over a prime field. In this work we present a characterization of a class of optimal three-weight cyclic codes of dimension 3 over any finite field. The codes under this characterization are, indeed, optimal in the sense that their lengths reach the Griesmer lower bound for linear codes. Consequently, these codes reach, simultaneously, the best possible coding capacity, and also the best possible capabilities of error detection and correction for linear codes. But because they are cyclic in nature, they also possess a rich algebraic structure that can be utilized in a variety of ways, particularly, in the design of very efficient coding and decoding algorithms. What is also worth pointing out, is the simplicity of the necessary and sufficient numerical conditions that characterize our class of optimal three-weight cyclic codes. As we already pointed out, it is a hard problem to find this kind of characterizations. However, for this particular case the fundamental tool that allowed us to find our characterization was the characterization for all two-weight irreducible cyclic codes that was introduced by \cite{seis}. Lastly, another feature about the codes in this class, is that their duals seem to have always the same parameters as the best known linear codes.
\end{abstract}

\begin{keyword}
Cyclic codes \sep weight distribution \sep Gaussian sums \sep Griesmer lower bound
\MSC 11T71 \sep 11T55 \sep 12E20
\end{keyword}

\end{frontmatter}

\section{Introduction}
The weight distribution of a code is important because it plays a significant role in determining their capabilities of error detection and correction. For cyclic codes this problem gains greater interest due mainly to the fact that they possess a rich algebraic structure. However, as was pointed out by \cite{cerobis1}, the problem of determining the weight distributions of families of cyclic codes is, in general, notoriously difficult. An even harder problem is to find characterizations for these families of cyclic codes in terms of their weight distributions. In fact, very few characterizations of this kind are known. One of the first and most relevant efforts in that direction is the work of \cite{seis}, where simple necessary and sufficient numerical conditions for an irreducible cyclic code to have at most two weights, are presented. Along the same lines, there also exists a set of characterizations, for the one-weight irreducible cyclic codes that was introduced in \cite{ocho}. In the case of reducible cyclic codes a characterization for the two-weight projective cyclic codes, was recently presented by \cite{cerobis2}.   

On the other hand, it is also well known that cyclic codes with few weights have a great practical importance in coding theory and cryptography, and this is so because they are useful in the design of frequency hopping sequences and in the development of secret sharing schemes. Some work has already been done in relation with irreducible and reducible two-weight cyclic codes (see for example, \cite{seis}, \cite{nueve} and \cite{cerobis2}), however we believe that for the particular case of reducible three-weight cyclic codes, whose duals have two zeros, a more in-depth study is possible. For example, cyclic codes having three weights have been studied by several authors (see for example, \cite{diez}, \cite{once} and \cite{tres}), however most of the efforts in that direction have been focused on cyclic codes over a prime field. In this work we present a characterization of a class of optimal three-weight cyclic codes of dimension 3 over any finite field. The codes in this class are, indeed, optimal in the sense that their lengths reach the Griesmer lower bound for linear codes. Therefore, in addition to the rich algebraic structure that is intrinsically associated with all cyclic codes, our codes also reach the best possible coding capacity, and also they have the best possible capabilities for error detection and correction for linear codes. As a further result of this work, we also find the parameters for the dual code of any cyclic code in our characterized class. In fact, throughout several studied examples, it seems that such dual codes always have the same parameters as the best known linear codes.

\begin{center}
TABLE I \\
{\em Weight distribution of ${\cal C}_{((q+1)e_1,e_2)}$.}
\end{center}
\begin{center}
\begin{tabular}{|c|c|} \hline
{\bf Weight} & $\;$ {\bf Frequency} $\;$\\ \hline \hline
0 & 1 \\ \hline
$q(q-1)-1$ & $(q-1)(q^2-1)$ \\ \hline
$q(q-1)$ & $q^2-1$ \\ \hline
$q^2-1$ & $q-1$ \\ \hline 
\end{tabular}
\end{center}

In order to provide a detailed explanation of what is the main result of this work, let $q$ be the power of a prime number, and also let $\gamma$ be a fixed primitive element of $\bbbf_{q^2}$. For any integer $a$, denote by $h_a(x) \in \bbbf_{q}[x]$ the minimal polynomial of $\gamma^{-a}$. With this notation in mind, the following result gives a full description for the weight distribution of a class of optimal three-weight cyclic codes of length $q^2-1$, dimension 3 over the finite field $\bbbf_{q}$.

\begin{theorem}\label{teouno}
For any two integers $e_1$ and $e_2$, let ${\cal C}_{((q+1)e_1,e_2)}$ be the cyclic code, over $\bbbf_{q}$, whose parity-check polynomial is $h_{(q+1)e_1}(x)h_{e_2}(x)$. Thus, if $\gcd(q-1,2e_1-e_2)=1$ and $\gcd(q+1,e_2)=1$ then 

\begin{enumerate}
\item[{\em (A)}] $\deg(h_{(q+1)e_1}(x))=1$ and $\deg(h_{e_2}(x))=2$. In addition, $h_{(q+1)e_1}(x)$ and $h_{e_2}(x)$ are the parity-check polynomials of two different one-weight cyclic codes of length $q^2-1$, whose nonzero weights are, respectively, $q^2-1$ and $q(q-1)$.
\\ 
\item[{\em (B)}] ${\cal C}_{((q+1)e_1,e_2)}$ is an optimal three-weight $[q^2-1,3,q(q-1)-1]$ cyclic code over $\bbbf_{q}$, with the weight distribution given in Table I. In addition, if $B_j$, with $0<j\leq q^2-1$, is the number of words of weight $j$ in the dual code of ${\cal C}_{((q+1)e_1,e_2)}$, then $B_1=B_2=0$ and 

$$B_3=\frac{(q^2-3)(q^2-1)(q-2)(q-1)}{6} \; .$$

\noindent
Therefore, if $q>2$, then the dual code of ${\cal C}_{((q+1)e_1,e_2)}$ is a single-error-correcting cyclic code with parameters $[q^2-1,q^2-1-3,3]$.
\end{enumerate}
\end{theorem}

For the particular case when $q$ is an even integer, and in connection with the class cyclic codes given by Theorem \ref{teouno}, the following result was recently presented in \cite{tresbis}. 

\begin{theorem}\label{teocuatro}
With our notation, suppose that $q$ is an even integer. Then ${\cal C}_{(q+1,q-1)}$ is a three-weight $[q^2-1,3,q(q-1)-1]$ cyclic code over $\bbbf_{q}$, with the weight distribution given in Table I.
\end{theorem} 

Now, since $q$ is an even integer, clearly $\gcd(q-1,2(1)-(q-1))=\gcd(q-1,2)=1$ and $\gcd(q+1,q-1)=1$. Therefore, it is interesting to observe that the family of codes given by the previous theorem are completely contained in the class of cyclic codes studied by Theorem \ref{teouno}. Therefore our main result not only extends the family of codes in Theorem 2, but also it extends the previous result to cyclic codes over {\em any} finite field. In fact, this is not all that can be said because, as will be outlined below, all cyclic codes, over $\bbbf_{q}$, of length $q^2-1$, whose weight distributions are given in Table I, satisfy the two easy-to-check conditions in Theorem \ref{teouno}. In other words, Theorem \ref{teouno} is, indeed, a characterization of a class of optimal three-weight cyclic codes of dimension 3 over any finite field. As we already pointed out, it is a hard problem to find this kind of characterizations. However, for this particular case the fundamental tool that allowed us to find our characterization was, as will be shown later, the characterization for all two-weight irreducible cyclic codes that was introduced by \cite{seis}.  

This work is organized as follows: In Section 2 we establish our notation and recall the definition for the Gaussian sums. Section 3 is devoted to recalling the Griesmer lower bound, and also to presenting a result that will allow us to conclude that the class of codes in Theorem \ref{teouno} are optimal in the sense that their lengths reach such lower bound. In Section 4 we will study a kind of exponential sums that is important in order to determine the weights, and their corresponding frequencies, of the codes in Theorem \ref{teouno}. In fact, for this kind of exponential sums, we are going to find simple necessary and sufficient numerical conditions in order that the evaluation of any exponential sum of such kind is exactly equal to one. In Section 5 we use the definitions and results of the previous sections in order to present a formal proof of Theorem \ref{teouno}. After this, we will analyze the two easy-to-check conditions of Theorem 1 in order to give an explicit formula for the number of cyclic codes that satisfy such conditions. In addition we include, at the end of this section, some examples of Theorem \ref{teouno} as well as some examples of such explicit formula. In Section 6 we will prove that the two easy-to-check sufficient numerical conditions in Theorem 1 are also the necessary conditions. Finally Section 7 will be devoted to present ours conclusions.

\section{Notation and some definitions}

First of all we set for this section and for the rest of this work, the following:

\medskip

\noindent
{\bf Notation.} By using $q$ we will denote the power of a prime number, whereas by using $\gamma$ we will denote a fixed primitive element of $\bbbf_{q^2}$. We are going to fix $\delta:=\gamma^{q+1}$, and consequently note that $\delta$ is a fixed primitive element of $\bbbf_{q}$. For any integer $a$, the polynomial $h_a(x) \in \bbbf_{q}[x]$ will denote the minimal polynomial of $\gamma^{-a}$. In addition, we will denote by ``$\mbox{Tr}_{\bbbf_{q^2}/\bbbf_q}$" the trace mapping from $\bbbf_{q^2}$ to $\bbbf_q$.

An important tool for this work are the so-called Gaussian sums. Thus, in order to recall such tool, let $\psi$ be a multiplicative and $\chi$ an additive character of a finite field $F$. Then the {\em Gaussian sum}, $G_{F}(\psi,\chi)$, of the characters $\psi$ and $\chi$ over $F$ is defined by

$$G_{F}(\psi,\chi) := \sum_{c \in F^*} \psi(c)\chi(c) \; .$$

There are several other results related to Gaussian sums that will be important for this work. Fortunately, these results are perfectly well explained in Chapter 5 of \cite{cuatro}.

\section{The Griesmer lower bound}

Let $n_q(k,d)$ be the minimum length $n$ for which an $[n,k,d]$ linear code, over $\bbbf_{q}$, exists. Given the values of $q$, $k$ and $d$, a central problem of coding theory is to determine the actual value of $n_q(k,d)$. A well-known lower bound (see \cite{uno} and \cite{siete}) for $n_q(k,d)$ is

\begin{theorem}\label{teodos}
(Griesmer bound) With the previous notation,

\[n_q(k,d) \geq \sum_{i=0}^{k-1} \left \lceil \frac{d}{q^i} \right \rceil \; .\]
\end{theorem} 

With the aid of the previous lower bound, we now present the following result.

\begin{lemma}\label{lemauno}
Suppose that ${\cal C}$ is a $[q^2-1,3,q(q-1)-1]$ linear code over $\bbbf_{q}$. Then ${\cal C}$ is an optimal linear code in the sense that its length reaches the lower bound in previous theorem.
\end{lemma} 

\begin{proof} 
By means of a direct application of the Griesmer lower bound, we have

\begin{eqnarray}
&&\left \lceil \frac{q(q-1)-1}{q^0} \right \rceil + \left \lceil \frac{q(q-1)-1}{q} \right \rceil + \left \lceil \frac{q(q-1)-1}{q^2} \right \rceil \nonumber \\
&=& [(q-1)q-1]+[q-1]+1=q^2-1 \; . \nonumber
\end{eqnarray}
\end{proof}

\section{A class of exponential sums}

It is well known that the weight distribution of some cyclic codes can be obtained by means of the evaluation of some exponential sums. This is particularly true for the class of cyclic codes that we are interested in. The following result goes along these lines.

\begin{lemma}\label{lemados}
Let $\chi'$ and $\chi$ be respectively the canonical additive characters of $\bbbf_{q^2}$ and $\bbbf_{q}$. For any integers $e_1$ and $e_2$, and for all $a,b \in \bbbf_{q^2}$, consider the sums

\[S_{(e_1,e_2)}(a,b) := \sum_{x \in \bbbf_{q^2}^*} \chi'(a x^{(q+1)e_1}+b x^{e_2}) \; .\]

\noindent
If $a^q+a \neq 0$, $b \neq 0$ and $\gcd(q+1,e_2)=1$, then

\[S_{(e_1,e_2)}(a,b)=-{\displaystyle \sum_{z \in \bbbf_{q}^*} \sum_{x \in \bbbf_{q}^*}\chi(z+(a^q+a)x^{e_1}+z^{-1}b^{q+1}x^{e_2})} \; .\]
\end{lemma} 

\begin{proof}
Recalling that $\delta:=\gamma^{q+1}$ we have

\begin{eqnarray}
S_{(e_1,e_2)}(a,b)&=&\sum_{i=0}^{q-2}\chi'(a\delta^{ie_1})\sum_{w \in \gamma^i \langle \gamma^{q-1} \rangle}\chi'(bw^{e_2}) \nonumber \\
&=&\sum_{i=0}^{q-2}\chi(\mbox{Tr}_{\bbbf_{q^2}/\bbbf_{q}}(a)\delta^{ie_1})\sum_{w \in \gamma^i \langle \gamma^{q-1} \rangle}\chi'(bw^{e_2}) \; , \nonumber
\end{eqnarray}

\noindent
and, since $\gcd(q+1,e_2)=1$, we have 

\[\sum_{w \in \gamma^i \langle \gamma^{q-1} \rangle}\chi'(bw^{e_2})=\sum_{w \in \gamma^{i e_2} \langle \gamma^{q-1} \rangle}\chi'(bw)=\frac{1}{q-1}\sum_{w \in \bbbf_{q^2}^*}\chi'(b \gamma^{i e_2} w^{q-1}) \; .\]

Let $\widehat{\bbbf}_{q^2}$ and $\widehat{\bbbf}_{q}$ be respectively the multiplicative character groups of $\bbbf_{q^2}$ and $\bbbf_{q}$. Now, if $\mbox{N}$ is the norm mapping from $\bbbf_{q^2}$ to $\bbbf_{q}$ and $H$ is the subgroup of order $q-1$ of $\widehat{\bbbf}_{q^2}$, then note that $H=\{\psi \circ \mbox{N} \: | \: \psi \in \widehat{\bbbf}_{q} \}$ (that is, $H$ is nothing but the ``lift" of $\widehat{\bbbf}_{q}$ to $\bbbf_{q^2}$). Therefore, owing to Theorem 5.30 (p. 217) in \cite{cuatro}, we have

\begin{eqnarray}
\sum_{w \in \bbbf_{q^2}^*}\chi'(b \gamma^{i e_2} w^{q-1})&=&\sum_{\psi \in \widehat{\bbbf}_{q}} G_{\bbbf_{q^2}}(\bar{\psi} \circ \mbox{N},\chi') \psi(\mbox{N}(b \gamma^{i e_2})) \nonumber \\
&=& -\sum_{\psi \in \widehat{\bbbf}_{q}} G_{\bbbf_{q}}(\bar{\psi},\chi)^2 \psi(\mbox{N}(b \gamma^{i e_2})) \; , \nonumber
\end{eqnarray}

\noindent
where the last equality arises due to the Davenport-Hasse theorem (Theorem 5.14 (p. 197) in \cite{cuatro}). In consequence, since $\gamma^{i(q+1)}=\mbox{N}(\gamma^i)=\mbox{N}(\gamma)^i=\delta^{i}$ and $ \langle \delta \rangle=\bbbf_{q}^*$, we have 

\begin{equation}\label{equno}
S_{(e_1,e_2)}(a,b)=-\frac{1}{q-1}\sum_{x \in \bbbf_{q}^*}\chi((a^q+a)x^{e_1}) \sum_{\psi \in \widehat{\bbbf}_{q}} G_{\bbbf_{q}}(\bar{\psi},\chi)^2 \psi(b^{q+1}x^{e_2}) \; .
\end{equation}

On the other hand, by using the Fourier expansion of the restriction of $\chi$ to $\bbbf_{q}^*$ in terms of the multiplicative characters of $\bbbf_{q}$, we have that for all $x,z \in \bbbf_{q}^*$:

\[\chi(z^{-1}b^{q+1}x^{e_2})=\frac{1}{q-1}\sum_{\psi \in \widehat{\bbbf}_{q}} G_{\bbbf_{q}}(\bar{\psi},\chi) \bar{\psi}(z) \psi(b^{q+1}x^{e_2}) \; ,\]

\noindent
and by multiplying both sides of the preceding equation by $\chi(z)$ and by summing we obtain

\[\sum_{z \in \bbbf_{q}^*}\chi(z+z^{-1}b^{q+1}x^{e_2})=\frac{1}{q-1}\sum_{\psi \in \widehat{\bbbf}_{q}} G_{\bbbf_{q}}(\bar{\psi},\chi)^2 \psi(b^{q+1}x^{e_2}) \; .\]

\noindent
Finally, by substituting the previous equation in (\ref{equno}) we obtain the desired result.
\end{proof}

\begin{remark}\label{rmuno}
It is worth pointing out that an important part of the previous proof was inspired by the proof of Theorem 2.8 in \cite{cinco}. 
\end{remark}

Now we are going to analyze a kind of exponential sums that are constructed by means of those exponential sums studied in the previous lemma. In fact, in what follows we are going to find simple necessary and sufficient numerical conditions in order that the evaluation of any exponential sum of such kind is exactly equal to one. 

\begin{lemma}\label{lematres}
With the same notation and hypothesis as in the previous lemma, consider now the sums of the form:

\[T_{(e_1,e_2)}(a,b) := \sum_{y \in \bbbf_{q}^*} S_{(e_1,e_2)}(ya,yb) \; .\]

\noindent
Then $\gcd(q-1,2e_1-e_2)=1$ if and only if $T_{(e_1,e_2)}(a,b)=1$.
\end{lemma}

\begin{proof}
Suppose that $\gcd(q-1,2e_1-e_2)=1$, then, from previous lemma and since $y^{q+1}=y^2$ for all $y \in \bbbf_{q}^*$, we have 

\begin{equation}\label{eqdos}
T_{(e_1,e_2)}(a,b)=-{\displaystyle \sum_{y \in \bbbf_{q}^*} \sum_{z \in \bbbf_{q}^*} \sum_{x \in \bbbf_{q}^*}\chi(z+(a^q+a)x^{e_1}y+z^{-1}b^{q+1}x^{e_2}y^2)} \; .
\end{equation}

First suppose that $q$ is even. Then by Theorem 5.34 (p. 218) in \cite{cuatro} we know that, for all $\rho_0, \rho_1, \rho_2 \in \bbbf_{q}$, 

\[\sum_{y \in \bbbf_{q}}\chi(\rho_0+\rho_1y+\rho_2y^2)=
\left\{ \begin{array}{cl}
		\chi(\rho_0)q & \mbox{ if $\rho_2+\rho_1^2=0$,} \\
		0 & \mbox{ otherwise.}
	\end{array}
\right . \]

\noindent
Therefore

\begin{eqnarray}\label{eqtres}
T_{(e_1,e_2)}(a,b)&=&1-q-\sum_{x \in \bbbf_{q}^*} \sum_{z \in \bbbf_{q}^*} \sum_{y \in \bbbf_{q}}\chi(z+(a^q+a)x^{e_1}y+z^{-1}b^{q+1}x^{e_2}y^2)  \nonumber \\
&=&1-q-q\sum_{x \in \bbbf_{q}^*} \chi((a^q+a)^{-2}b^{q+1}x^{e_2-2e_1}) \; , 
\end{eqnarray}

\noindent
and because $\gcd(q-1,2e_1-e_2)=1$, we conclude that $T_{(e_1,e_2)}(a,b)=1$. 

For the case when $q$ is odd, we first prove that $T_{(1,1)}(a,b)=1$. Thus, by making the variable substitution $x \mapsto b^{-(q+1)}y^{-2}x$, in the inner summation of (\ref{eqdos}) (recall $e_1=e_2=1$), we get

\[T_{(1,1)}(a,b)=-\sum_{z \in \bbbf_{q}^*}\chi(z) \sum_{x \in \bbbf_{q}^*}\chi(z^{-1}x) \sum_{y \in \bbbf_{q}^*} \chi(b^{-(q+1)}(a^q+a)x y^{-1})=1 \; .\]

Now suppose that $q$ is odd and that $e_1$ and $e_2$ are any two integers. Then by Theorem 5.33 in \cite{cuatro} we know that, for all $\rho_0, \rho_1, \rho_2 \in \bbbf_{q}$ with $\rho_2 \neq 0$, 

\[\sum_{y \in \bbbf_{q}^*}\chi(\rho_0+\rho_1y+\rho_2y^2)=\chi(\rho_0-\rho_1^2(4\rho_2)^{-1})\eta(\rho_2)G_{\bbbf_{q}}(\eta,\chi)-\chi(\rho_0) \; ,\]

\noindent
where $\eta$ is the quadratic character of $\bbbf_{q}$. Therefore, from (\ref{eqdos}), we have  

\begin{equation}\label{eqcuatro}
T_{(e_1,e_2)}(a,b)=-\!\!\sum_{x \in \bbbf_{q}^*} \sum_{z \in \bbbf_{q}^*}\!\chi(z-cx^{2e_1-e_2}z)\eta(z^{-1}N_bx^{e_2})G_{\bbbf_{q}}(\eta,\chi)-\chi(z) \: ,
\end{equation}

\noindent
where $T_a:=a^q+a$, $N_b:=b^{q+1}$ and $c:=4^{-1}T_a^2 N_b^{-1}$. But $q$ is odd and $\gcd(q+1,e_2)=1$, therefore both $e_2$ and $2e_1-e_2$ must be odd integers. Consequently, since $\gcd(q-1,2e_1-e_2)=1$, there must exist an odd integer $r$ such that $(2e_1-e_2)r \equiv 1 \pmod{q-1}$. Therefore, by applying the variable substitution $x \mapsto x^r$ in the previous equality, we now have 

\[T_{(e_1,e_2)}(a,b)=-\!\!\sum_{x \in \bbbf_{q}^*} \sum_{z \in \bbbf_{q}^*}\!\chi(z-cxz)\eta(z^{-1}N_bx^{re_2})G_{\bbbf_{q}}(\eta,\chi)-\chi(z) \; ,\]

\noindent
and since $e_2$ and $r$ are both odd integers, clearly $\eta(z^{-1}N_b x^{re_2})=\eta(z^{-1}N_b x)$. Therefore

\begin{eqnarray}
T_{(e_1,e_2)}(a,b)\!\!&=&-\!\!\sum_{x \in \bbbf_{q}^*} \sum_{z \in \bbbf_{q}^*}\!\chi(z-cxz)\eta(z^{-1}N_bx)G_{\bbbf_{q}}(\eta,\chi)-\chi(z) \nonumber \\
\!\!&=& \; T_{(1,1)}(a,b)=1 \; . \nonumber
\end{eqnarray}

For the proof of the converse, suppose that $\gcd(q-1,2e_1-e_2)=d>1$. Thus, again by first supposing that $q$ is even, we have from (\ref{eqtres}) that

\begin{eqnarray}
T_{(e_1,e_2)}(a,b)&=&1-q-qd\sum_{x \in \langle \delta^{d} \rangle} \chi((a^q+a)^{-2}b^{q+1}x)  \nonumber \\
&=&1-q-qdt \; , \nonumber
\end{eqnarray}

\noindent
for some integer $t$ (recall that $\chi(w) =\pm 1$, for all $w \in \bbbf_{q}$), and since $d>1$, we have that $dt \neq -1$. Therefore $T_{(e_1,e_2)}(a,b) \neq 1$.

Finally, suppose that $\gcd(q-1,2e_1-e_2)=d>1$ and that $q$ is odd. In this case note that $e_2$ and $d$ are also odd integers. Thus, since $\eta(x^{e_2})=\eta(x^{2e_1-e_2})$ for all $x \in \bbbf_{q}^*$, we have from (\ref{eqcuatro}):

\begin{eqnarray}
T_{(e_1,e_2)}(a,b)\!\!\!&=&-\!\!\!\sum_{x \in \bbbf_{q}^*} \sum_{z \in \bbbf_{q}^*}\!\chi(z-cx^{2e_1-e_2}z)\eta(z^{-1}N_bx^{2e_1-e_2})G_{\bbbf_{q}}(\eta,\chi)-\chi(z)  \nonumber \\
\!\!\!&=&\!\!\!1-q-G_{\bbbf_{q}}(\eta,\chi)\sum_{x \in \bbbf_{q}^*} \eta(N_b x) \sum_{z \in \bbbf_{q}^*}\chi((1-cx^d)z) \eta(z^{-1}) \; , \nonumber
\end{eqnarray}

\noindent
because  $\gcd(q-1,2e_1-e_2)=d$ and $\eta(x^d)=\eta(x)$. 

Now, if ${\cal B}:=\{ x \in \bbbf_{q}^* \: | \: x^d=c^{-1} \}$, then observe that $|{\cal B}|=0$ or $|{\cal B}|=d$, and

\[\sum_{x \in {\cal B}} \sum_{z \in \bbbf_{q}^*}\chi((1-cx^d)z) \eta(z^{-1})=\sum_{x \in {\cal B}} \sum_{z \in \bbbf_{q}^*}\eta(z^{-1})=0 \; .\]

\noindent
Therefore

\begin{eqnarray}
T_{(e_1,e_2)}(a,b)&=&1-q-G_{\bbbf_{q}}(\eta,\chi)\sum_{x \in \bbbf_{q}^* \setminus {\cal B}} \eta(N_b x) \sum_{z \in \bbbf_{q}^*}\bar{\chi}((cx^d-1)z) \bar{\eta}(z) \nonumber \\
&=&1-q-G_{\bbbf_{q}}(\eta,\chi)\sum_{x \in \bbbf_{q}^* \setminus {\cal B}} \eta(N_b x)  G_{\bbbf_{q}}(\bar{\eta},\bar{\chi}) \eta(cx^d-1) \; , \nonumber
\end{eqnarray}

\noindent
where the last equality arises due to the Part (i) of Theorem 5.12 (p. 193) in \cite{cuatro}. But $G_{\bbbf_{q}}(\eta,\chi)G_{\bbbf_{q}}(\bar{\eta},\bar{\chi})=q$, therefore

\begin{eqnarray}
T_{(e_1,e_2)}(a,b)&=&1-q-q\sum_{x \in \bbbf_{q}^* \setminus {\cal B}} \eta(N_b x) \eta(cx^d-1) \nonumber \\
&=&1-q-q\sum_{x \in \bbbf_{q}^* \setminus {\cal B}} \eta(x) \eta(x^d-c^{-1}) \; , \nonumber
\end{eqnarray}

\noindent
because $\eta(N_b)=\eta(c^{-1})$. Now, if ${\cal D} := \{ \delta^i \: | \: 0 \leq i < \frac{q-1}{d} \}$, then note that 

\[|{\cal D} \cap {\cal B}|=
\left\{ \begin{array}{cl}
		0 & \mbox{ if $|{\cal B}|=0$,} \\
		1 & \mbox{ if $|{\cal B}|=d$}.
			\end{array}
\right . \]

\noindent
Thus, 

\[T_{(e_1,e_2)}(a,b)=1-q-q\sum_{x \in {\cal D} \setminus ({\cal D} \cap {\cal B})} \sum_{j=0}^{d-1} \eta(x \delta^{j\frac{q-1}{d}}) \eta((x \delta^{j\frac{q-1}{d}})^d-c^{-1})\; ,\]

\noindent
but since $\frac{q-1}{d}$ is even, we have $\eta(x \delta^{j\frac{q-1}{d}})=\eta(x)$ and clearly $(x \delta^{j\frac{q-1}{d}})^d=x^d$. Therefore 

\begin{eqnarray}
T_{(e_1,e_2)}(a,b)&=&1-q-q\sum_{x \in {\cal D} \setminus ({\cal D} \cap {\cal B})} \sum_{j=0}^{d-1} \eta(x) \eta(x^d-c^{-1}) \nonumber \\
&=&1-q-qd\sum_{x \in {\cal D} \setminus ({\cal D} \cap {\cal B})} \eta(x) \eta(x^d-c^{-1}) \nonumber \\
&=&1-q-qdt \; , \nonumber
\end{eqnarray}

\noindent
for some integer $t$ (recall that $\eta(w) =\pm 1$, for all $w \in \bbbf_{q}^*$), and since $d>1$, we have that $dt \neq -1$. Therefore $T_{(e_1,e_2)}(a,b) \neq 1$.
\end{proof}

\begin{remark}\label{rmtres}
Let $\bbbf_{q}$ be any finite field of odd characteristic and let $\eta$ be the quadratic character of $\bbbf_{q}$. If $\rho$ is any nonzero element of $\bbbf_{q}$ then note that, as a consequence of the previous proof, it is possible to conclude that

\[|\{ x \in \bbbf_{q}^* \setminus \{\rho\} \; | \; \eta(x^2 - \rho x)=1 \}|=\frac{q-1}{2}-1 \; . \]
\end{remark}
  
We end this section of preliminary results by presenting the following

\begin{corollary}\label{coruno}
With the same notation as in the previous lemma, let $a,b \in \bbbf_{q^2}$. If $\gcd(q-1,2e_1-e_2)=1$ and $\gcd(q+1,e_2)=1$, then

\[T_{(e_1,e_2)}(a,b)=
\left\{ \begin{array}{cl}
		(q-1)(q^2-1) & \mbox{ if $\;\;\;a=0$ and $b=0$,} \\
		-(q^2-1)     & \mbox{ if $a^q+a \neq 0$ and $b=0$,} \\
		-(q-1)       & \mbox{ if $a^q+a=0$ and $b \neq 0$,} \\
		 1           & \mbox{ if $a^q+a \neq 0$ and $b \neq 0$}.
			\end{array}
\right . \]
\end{corollary}

\begin{proof}
Clearly $T_{(e_1,e_2)}(0,0)=(q-1)(q^2-1)$ and, if $\mbox{Tr}_{\bbbf_{q^2}/\bbbf_{q}}(a) \neq 0$ and $b=0$, then 

\begin{eqnarray}
T_{(e_1,e_2)}(a,0)&=&\sum_{y \in \bbbf_{q}^*} \sum_{x \in \bbbf_{q^2}^*}\chi'(a x^{(q+1)e_1}y) \nonumber \\
&=&\sum_{x \in \bbbf_{q^2}^*} \sum_{y \in \bbbf_{q}^*}\chi(\mbox{Tr}_{\bbbf_{q^2}/\bbbf_{q}}(a) x^{(q+1)e_1}y)=-(q^2-1) \; . \nonumber
\end{eqnarray}

\noindent
On the other hand, if $a^q+a=0$ and $b \neq 0$, then 

\begin{eqnarray}
T_{(e_1,e_2)}(a,b)&=&\sum_{y \in \bbbf_{q}^*} \sum_{x \in \bbbf_{q^2}^*} \chi(\mbox{Tr}_{\bbbf_{q^2}/\bbbf_{q}}(a) x^{(q+1)e_1}y)\chi'(b x^{e_2}y) \nonumber \\
&=&\sum_{x \in \bbbf_{q^2}^*} \sum_{y \in \bbbf_{q}^*} \chi(\mbox{Tr}_{\bbbf_{q^2}/\bbbf_{q}}(b x^{e_2})y) \; \nonumber \\
&=&\sum_{i=0}^q \sum_{x \in \bbbf_{q}^*} \sum_{y \in \bbbf_{q}^*} \chi(x^{e_2}\mbox{Tr}_{\bbbf_{q^2}/\bbbf_{q}}(b \gamma^{i e_2})y) \;,\nonumber
\end{eqnarray}

\noindent
but, since $\gcd(q+1,e_2)=1$, we have

\begin{eqnarray}
T_{(e_1,e_2)}(a,b)&=&\sum_{i=0}^q \sum_{x \in \bbbf_{q}^*} \sum_{y \in \bbbf_{q}^*} \chi(x^{e_2}\mbox{Tr}_{\bbbf_{q^2}/\bbbf_{q}}(b \gamma^{i})y)  \nonumber \\
&=&\sum_{x \in \bbbf_{q}^*} \sum_{y \in \bbbf_{q}^*} \chi(0) + q\sum_{x \in \bbbf_{q}^*} \sum_{y \in \bbbf_{q}^*} \chi(y) \nonumber \\
&=&(q-1)^2 -q(q-1)=-(q-1)   \;.\nonumber
\end{eqnarray}

\noindent
Finally, the proof of the last case comes from previous lemma.
\end{proof}

\section{Formal proof of Theorem \ref{teouno}}

We are now able to present a formal proof of Theorem \ref{teouno}.

\begin{proof}
Part {(A)}: Clearly $(q+1)e_1q \equiv (q+1)e_1 \pmod{q^2-1}$ and, due to the fact that $\gcd(q+1,e_2)=1$, we have $e_2q \not\equiv e_2 \pmod{q^2-1}$, therefore $\deg(h_{(q+1)e_1}(x))=1$ and $\deg(h_{e_2}(x))=2$. Note that if ${\cal C}_{((q+1)e_1)}$ and ${\cal C'}_{((q+1)e_1)}$ are the cyclic codes with the same parity-check polynomial $h_{(q+1)e_1}(x)$, and whose lengths are, respectively, $q^2-1$ and $q-1$, then the weights of all codewords of these two codes differ just by the constant factor $q+1$. Now by using the set of characterizations, for the one-weight irreducible cyclic codes, that was introduced in Theorem 11 of \cite{ocho}, and since $\gcd(\frac{q^1-1}{q-1},(q+1)e_1)=1$, we conclude that ${\cal C'}_{((q+1)e_1)}$ is a cyclic code of length $q-1$, whose nonzero weight is $q-1$. Therefore the nonzero weight of ${\cal C}_{((q+1)e_1)}$ is $q^2-1$. On the other hand, because $\gcd(\frac{q^2-1}{q-1},e_2)=1$, we can conclude, in a similar manner, that $h_{e_2}(x)$ is the parity-check polynomial of a one-weight cyclic code of length $q^2-1$, whose nonzero weight is $q(q-1)$.

Part {(B)}: Clearly, the cyclic code ${\cal C}_{((q+1)e_1,e_2)}$ has length $q^2-1$ and its dimension is $3$ due to Part {(A)}. Let ${\cal A}$ be a fixed subset of $\bbbf_{q^2}^*$ in such a way that $|{\cal A}|=q-1$ and $\{\mbox{Tr}_{\bbbf_{q^2}/\bbbf_q}(a) \:|\: a \in {\cal A}\}=\bbbf_{q}^*$ (observe that if $q$ is even, then the subset ${\cal A}$ must be different from $\bbbf_{q}^*$). Now, for each $a \in {\cal A} \cup \{0\}$ and $ b \in \bbbf_{q^2}$, we define $c(q^2-1,e_1,e_2,a,b)$ as the vector of length $q^2-1$ over $\bbbf_q$, which is given by:

$$(\mbox{Tr}_{\bbbf_{q^2}/\bbbf_q}(a(\gamma^{(q+1)e_1})^i+b(\gamma^{e_2})^i))_{i=0}^{q^2-2}\; .$$

\noindent
Thanks to Delsarte's Theorem (\cite{cero}) it is well known that 

$${\cal C}_{((q+1)e_1,e_2)}=\{ c(q^2-1,e_1,e_2,a,b) \: | \: a \in {\cal A} \cup \{0\} \mbox{ and } b \in \bbbf_{q^2} \} \; .$$ 

\noindent
Thus the Hamming weight of any codeword $c(q^2-1,e_1,e_2,a,b)$, in our cyclic code ${\cal C}_{((q+1)e_1,e_2)}$, will be equal to $q^2-1-Z(a,b)$, where

$$Z(a,b)\!=\!\sharp\{\;i\; | \; \mbox{Tr}_{\bbbf_{q^2}/\bbbf_q}(a\gamma^{(q+1)e_1 i}+b\gamma^{e_2 i})=0, \: 0 \leq i < q^2-1 \} \:.$$

\noindent
If $\chi'$ and $\chi$ are, respectively, the canonical additive characters of $\bbbf_{q^2}$ and $\bbbf_{q}$, then

\begin{eqnarray}
Z(a,b)&=&\frac{1}{q}\sum_{i=0}^{q^2-2} \sum_{y \in \bbbf_q} \chi(\mbox{Tr}_{\bbbf_{q^2}/\bbbf_q}(y(a\gamma^{(q+1)e_1 i}+b\gamma^{e_2 i}))) \nonumber \\
&=&\frac{q^2-1}{q}+\frac{1}{q}\sum_{y \in \bbbf_q^*} \sum_{x \in \bbbf_{q^2}^*} \chi'(ya x^{(q+1)e_1}+yb x^{e_2})  \; , \nonumber
\end{eqnarray}   

\noindent
and, by using the notation of Lemma \ref{lematres}, we have

\begin{equation}\label{eqcinco}
Z(a,b)=\frac{q^2-1}{q}+\frac{1}{q} T_{(e_1,e_2)}(a,b) \; .
\end{equation}

\noindent
But $\gcd(q-1,2e_1-e_2)=1$ and $\gcd(q+1,e_2)=1$; therefore, after applying Corollary \ref{coruno}, we get

\[Z(a,b)=\left\{ \begin{array}{cl}
		q^2-1 & \mbox{ if $a=0$ and $b=0$,} \\
		0    & \mbox{ if $a \in {\cal A}$ and $b=0$,} \\
		q-1      & \mbox{ if $a=0$ and $b \neq 0$,} \\
		q       & \mbox{ if $a \in {\cal A}$ and $b \neq 0$}.
			\end{array}
\right . \]

\noindent
Consequently, the assertion about the weight distribution of ${\cal C}_{((q+1)e_1,e_2)}$ comes now from the fact that the Hamming weight of any codeword in ${\cal C}_{((q+1)e_1,e_2)}$ is equal to $q^2-1-Z(a,b)$, and also due to the fact that $|{\cal A}|=q-1$ and $|\bbbf_{q^2}^*|=q^2-1$.  

Lastly, ${\cal C}_{((q+1)e_1,e_2)}$ is an optimal cyclic code, due to Lemma \ref{lemauno}, and the assertion about the weights of the dual code ${\cal C}_{((q+1)e_1,e_2)}$ can now be proved by means of Table I and the first four identities of Pless (see, for example, pp. 259-260 in \cite{dos}).
\end{proof}

Due to the simplicity of the numerical conditions in Theorem \ref{teouno}, it is possible to compute the total number of different cyclic codes, over $\bbbf_{q}$, of length $q^2-1$ and dimension $3$, that satisfy such conditions. The following result goes in that direction.

\begin{theorem}\label{teotres}
With our notation, let ${\cal N}$ be the number of different cyclic codes, ${\cal C}_{((q+1)e_1,e_2)}$, of length $q^2-1$ and dimension $3$ that satisfy conditions in Theorem \ref{teouno}. Then

\begin{equation}\label{eqseis}
{\cal N} = \frac{\phi(q^2-1)(q-1)}{2} \; ,
\end{equation}

\noindent
where $\phi$ denotes the Euler $\phi$-function.
\end{theorem} 

\begin{proof}
Since $\deg(h_{e_2}(x))=2$, the total number, ${\cal N}_2$, of different minimal polynomials $h_{e_2}(x)$ that satisfy condition $\gcd(q+1,e_2)=1$ is ${\cal N}_2=\frac{\phi(q+1)(q-1)}{2}$. On the other hand, since $\deg(h_{(q+1)e_1}(x))=1$ we have that, for each integer $e_2$ that satisfies $\gcd(q+1,e_2)=1$, the total number, ${\cal N}_1$, of different minimal polynomials $h_{(q+1)e_1}(x)$ that satisfy condition $\gcd(q-1,2e_1-e_2)=1$ is

\begin{eqnarray}
{\cal N}_1&=&|\{0 \leq e_1 < (q-1) \:|\: \gcd(q-1,2e_1-e_2)=1 \}|  \nonumber \\
&&   \nonumber \\
&=&\left\{ \begin{array}{cl}
		\phi(q-1)  & \mbox{ if $q$ is even,} \\
		2\phi(q-1) & \mbox{ otherwise.}
	   \end{array}
\right .   \nonumber
\end{eqnarray}   

\noindent
But since ${\cal N}={\cal N}_1 {\cal N}_2$, the result now follows from the fact that

\[\phi(q+1)\phi(q-1)=\left\{ \begin{array}{cl}
		\phi(q^2-1)  & \mbox{ if $q$ is even,} \\
		\phi(q^2-1)/2 & \mbox{ otherwise.}
	   \end{array}
\right . \]
\end{proof}

The following are examples related to Theorem \ref{teouno} and Theorem \ref{teotres}.

\begin{example}\label{ejecero}
With the same notation as in Theorem \ref{teouno}, let $q=4$, $e_1=2$ and $e_2=6$. Then $\gcd(q-1,2e_1-e_2)=1$ and $\gcd(q+1,e_2)=1$. Therefore, by Theorem \ref{teouno}, we can be sure that ${\cal C}_{(10,6)}$ is an optimal three-weight cyclic code over $\bbbf_{4}$, of length 15, dimension 3 and weight enumerator polynomial 

\begin{equation}\label{eqsiete}
1+45z^{11}+15z^{12}+3z^{15} \; .
\end{equation}

\noindent
In addition, $B_1=B_2=0$ and $B_3=195$. In fact, the dual code of ${\cal C}_{(10,6)}$ is a $[15,12,3]$ cyclic code which, by the way, has the same parameters as the best known linear code, according to the tables of the best known linear codes maintained by Markus Grassl at http://www.codetables.de/.
\end{example}

\begin{example}\label{ejeuno}
Again, we take $q=4$. Then, owing to Theorem \ref{teotres}, the total number of different cyclic codes, over $\bbbf_{4}$, of length $15$ and dimension $3$ that satisfy conditions of Theorem \ref{teouno} is ${\cal N}=12$. In fact, these cyclic codes are ${\cal C}_{(0,1)}$, ${\cal C}_{(0,2)}$, ${\cal C}_{(0,7)}$, ${\cal C}_{(0,11)}$, ${\cal C}_{(5,1)}$, ${\cal C}_{(5,3)}$, ${\cal C}_{(5,6)}$, ${\cal C}_{(5,7)}$, ${\cal C}_{(10,2)}$, ${\cal C}_{(10,3)}$, ${\cal C}_{(10,6)}$ and  ${\cal C}_{(10,11)}$. Now, through a direct inspection it is interesting to note that all different cyclic codes over $\bbbf_{4}$ of length $15$, dimension $3$ and weight enumerator polynomial as in (\ref{eqsiete}), are exactly those listed before.
\end{example}

\section{Towards the characterization}

Through the last example in the previous section, it can be conjectured that the sufficient numerical conditions in Theorem \ref{teouno} are also the necessary conditions. In fact, this is the real situation and the following result gives us a formal proof of this conjecture.

\begin{theorem}\label{teocinco}
Let ${\cal C}$ be a cyclic code of length $q^2-1$ over a finite field $\bbbf_{q}$. Then, the weight distribution of ${\cal C}$ is given in Table I if and only if its dimension is 3 and there exist two integers, $e_1$ and $e_2$, in such a way that $h_{(q+1)e_1}(x)h_{e_2}(x)$ is the parity-check polynomial of ${\cal C}$, and the two integers satisfy $\gcd(q-1,2e_1-e_2)=1$ and $\gcd(q+1,e_2)=1$.\end{theorem} 

\begin{proof}
Suppose that ${\cal C}$ is a cyclic code of length $q^2-1$ over a finite field $\bbbf_{q}$, whose weight distribution is given in Table I. Through the sum of the frequencies of such table, it is easy to see that ${\cal C}$ must be a cyclic code of dimension 3. Consequently, the degree of the parity-check polynomial $h(x)$, of ${\cal C}$, must be equal to $3$. Since $(q^2-1) \nmid (q^3-1)$, the code ${\cal C}$ cannot be an irreducible cyclic code of length $q^2-1$ and dimension $3$. Therefore the parity-check polynomial $h(x)$ must be reducible. As was explained in the proof of Part (A) of Theorem \ref{teouno}, a cyclic code of length $q^2-1$ and dimension $1$ is just a one-weight irreducible cyclic code over $\bbbf_q$, whose nonzero weight is $q^2-1$. Thus, if $h(x)$ is the product of three polynomials of degree $1$, then ${\cal C}$ will correspond to the span of the union of three different one-weight irreducible cyclic codes (seeing them as three different subspaces of $\bbbf_q^{q^2-1}$), and therefore the frequency of the nonzero weight of $q^2-1$, in Table I, should be at least $3(q-1)$. Since this is not the situation for Table I, the polynomial $h(x)$ must be the product of two polynomials, one of them of degree $1$ and the other one of degree $2$. Seeing such polynomials as minimal polynomials over $\bbbf_{q^2}$, we have that there must exist two integers $e_1$ and $e_2$ in such a way that $h(x)=h_{(q+1)e_1}(x)h_{e_2}(x)$. 

Now, we are going to prove that $\gcd(q+1,e_2)=1$. Let ${\cal C}_{((q+1)e_1)}$ and ${\cal C}_{(e_2)}$ be the cyclic codes of length $q^2-1$ over $\bbbf_{q}$, whose parity-check polynomials are, respectively, $h_{(q+1)e_1}(x)$ and $h_{e_2}(x)$. If $\gcd(q+1,e_2)=u > 1$, then, due to the set characterizations for the one-weight irreducible cyclic codes, that was introduced in \cite{ocho}, ${\cal C}_{(e_2)}$ must have at least two nonzero weights. Since ${\cal C}_{((q+1)e_1)}$ is a one-weight irreducible cyclic code of length $q^2-1$, with nonzero weight $q^2-1$, and due to the fact that ${\cal C}$ is the span of the union of ${\cal C}_{((q+1)e_1)}$ and ${\cal C}_{(e_2)}$, we have that none of the nonzero weights of ${\cal C}_{(e_2)}$ can be equal to $q^2-1$. But in Table I there are just 3 different nonzero weights and one of them is equal to $q^2-1$. Thus the conclusion here is that if $\gcd(q+1,e_2)=u > 1$, then ${\cal C}_{(e_2)}$ will correspond to a two-weight irreducible cyclic code, over $\bbbf_{q}$, of length $q^2-1$ and dimension $2$, whose nonzero weights are $q(q-1)$ and $q(q-1)-1$. Fortunately for us, simple necessary and sufficient numerical conditions for an irreducible cyclic code to have at most two weights were presented in the remarkable work of \cite{seis}. Despite the fact that such characterization is just for all two-weight irreducible cyclic codes over a prime field, the authors provided all the required clues to extend their characterization to {\em any} finite field. Thus, taking into consideration these clues it is possible to obtain the following characterization for all the two-weight irreducible cyclic codes of length $q^2-1$ and dimension $2$ over any finite field (see Theorem 6 and its proof in \cite{nuevebis}).

\begin{center}
TABLE II \\
{\em Weight distribution of a two-weight code ${\cal C}_{(e)}$. \\}
Here $\varepsilon=\pm 1$ is determined by $r p^{s\theta} \equiv \varepsilon \pmod{u}$.
\end{center}
\begin{center}
\begin{tabular}{|c|c|} \hline
{\bf Weight} & $\;$ {\bf Frequency} $\;$\\ \hline \hline
0 & 1 \\ \hline
$\; \frac{q-1}{q}(q^2 - r \varepsilon p^{s\theta}) \;$ & $\frac{(q^2-1)(u-r)}{u}$ \\ \hline
$\; \frac{q-1}{q}(q^2 + (u-r) \varepsilon p^{s\theta}) \;$ & $\frac{(q^2-1)r}{u}$ \\ \hline
\end{tabular}
\end{center}

\begin{theorem}\label{teoseis}
Let $p$, $t$ and $q$ be positive integers in such a way that $p$ is a prime number and $q=p^t$. For any integer $e$, let ${\cal C}_{(e)}$ be the irreducible cyclic code, over $\bbbf_{q}$, of length $q^2-1$, whose parity-check polynomial is $h_e(x)$, and suppose that $\deg(h_e(x))=2$. For $u=\gcd(q+1,e)$, let $f$ and $s$ be the two integers in such a way that $2t=fs$, with $f:=\mbox{ord}_{u}(p)$ (that is, $f$ is {\em the multiplicative order of $p$ modulo $u$}). For a positive integer $x$, let $S_p(x)$ denote the sum of the $p$-digits of $x$. Define

\[\theta(u,p)=\frac{1}{p-1} \min \left\{ S_p\left(\frac{j(p^f-1)}{u}\right) \; | \; 1 \leq j < u \right\} \;,\]

\noindent
and fix $\theta=\theta(u,p)$. Then the irreducible cyclic code ${\cal C}_{(e)}$ has the weight distribution given in Table II if and only if $u>1$ and there exists a positive integer $r$ satisfying 

\begin{eqnarray}
&& r | (u-1) \nonumber \\
&& r p^{s\theta} \equiv \pm 1 \pmod{u}  \nonumber \\
&& r (u-r)=(u-1) p^{s(f-2\theta)} \; .  \nonumber
\end{eqnarray}
\end{theorem} 

Now, by observing Table II and by noting that $r$ and $u-r$ cannot be zero in the previous theorem, we have that the nonzero weights of any two-weight irreducible cyclic code of length $q^2-1$ and dimension $2$ can never be equal to $q(q-1)$. But this is a contradiction, because we already conclude that the nonzero weights of ${\cal C}_{(e_2)}$ are $q(q-1)$ and $q(q-1)-1$. Therefore ${\cal C}_{(e_2)}$ cannot be a two-weight irreducible cyclic code, and in consequence, $\gcd(q+1,e_2)=1$.

It remains to prove that $\gcd(q-1,2e_1-e_2)=1$. If $\gcd(q+1,e_2)=1$, then, once again, as was explained in proof of Part (A) of Theorem \ref{teouno}, ${\cal C}_{(e_2)}$ will correspond to a one-weight irreducible cyclic code of length $q^2-1$ and dimension $2$, whose nonzero weight is $q(q-1)$. Since the frequency of such nonzero weight is $q^2-1=|\bbbf_{q^2}^*|$, in Table I, we have that a codeword, $c$, in ${\cal C}$ will have Hamming weight $q(q-1)-1$ if and only if $c=c_1+c_2$, where $c_1$ and $c_2$ are, respectively, two nonzero codewords in ${\cal C}_{((q+1)e_1)}$ and ${\cal C}_{(e_2)}$. But if $c_1$ and $c_2$ are nonzero codewords in ${\cal C}_{((q+1)e_1)}$ and ${\cal C}_{(e_2)}$, then there must exist two finite field elements $a$ and $b$ in $\bbbf_{q^2}$, with $a^q+a \neq 0$, $b \neq 0$, in such a way that the number of zero entries, $Z(a,b)$, in codeword $c$, can be computed by means of (\ref{eqcinco}). Under these circumstances, codeword $c$ will have Hamming weight $q(q-1)-1$ if and only if $T_{(e_1,e_2)}(a,b)=1$, and due to Lemma \ref{lematres}, this can only be possible if and only if $\gcd(q-1,2e_1-e_2)=1$. 

Finally, the proof of the converse is just a part of the proof of Theorem \ref{teouno} that was already given in previous section.
\end{proof}

As a direct consequence of Theorems \ref{teotres} and \ref{teocinco}, we have the following result.

\begin{corollary}\label{cordos}
Let ${\cal N}$ be the number of different cyclic codes of length $q^2-1$, over $\bbbf_{q}$, whose weight distribution is given in Table I. Then ${\cal N}$ is given by (\ref{eqseis}).
\end{corollary}

\section{Conclusions}

In this work we presented a characterization of a class of optimal three-weight cyclic codes of length $q^2-1$ and dimension 3, over any finite field $\bbbf_{q}$. The codes under this characterization are, indeed, optimal in the sense that their lengths reach the Griesmer lower bound for linear codes. In addition, we also found the parameters for the dual code of any cyclic code in this class. In fact, throughout several studied examples, it seems that such dual codes have always the same parameters as the best known linear codes. As we saw in Example \ref{ejeuno}, it is easy to find all cyclic codes over a fixed finite field $\bbbf_{q}$ of length $q^2-1$ and dimension $3$ that satisfy the two conditions of Theorem \ref{teouno}. But due to Theorem \ref{teocinco} we can be sure that these cyclic codes will be all optimal three-weight cyclic codes of length $q^2-1$, whose weight distribution is given in Table I. As a complement of this work, we believe that it could be interesting the study of the family cyclic codes of length $q^2-1$, whose parity-check polynomial is in the form of $h(x)=h_{(q+1)e_1}(x)h_{e_2}(x)$, where the integers $e_1$ and $e_2$ satisfy $\gcd(q-1,2e_1-e_2) > 1$ and $\gcd(q+1,e_2)=1$.

\end{document}